\documentclass[twoside]{article}
\usepackage{jmlr2e}
 
\usepackage{natbib}
\usepackage{graphicx}
\usepackage{latexsym}
\usepackage{amssymb}
\usepackage{amsmath}
\usepackage{amsthm}
\usepackage{forloop}
\usepackage{subfig} 
\usepackage[paper=letterpaper,margin=1in]{geometry}
\usepackage[ruled,vlined,linesnumbered,norelsize]{algorithm2e}
\usepackage{url}
\usepackage[parfill]{parskip}

\newtheorem{theorem}{Theorem}[section]
\newtheorem{lemma}[theorem]{Lemma}

\newtheorem{corollary}[theorem]{Corollary}

\newtheorem{observation}[theorem]{Observation}

\makeatletter
\newtheorem*{rep@theorem}{\rep@title}
\newcommand{\newreptheorem}[2]{%
\newenvironment{rep#1}[1]{%
 \def\rep@title{#2 \ref{##1}}%
 \begin{rep@theorem}}%
 {\end{rep@theorem}}}
\makeatother

\newreptheorem{theorem}{Theorem}
\newreptheorem{reduction}{Reduction}
\newcommand{\defcal}[1]{\expandafter\newcommand\csname c#1\endcsname{{\mathcal{#1}}}}
\newcommand{\defbb}[1]{\expandafter\newcommand\csname b#1\endcsname{{\mathbb{#1}}}}
\newcounter{calBbCounter}
\forLoop{1}{26}{calBbCounter}{
    \edef\letter{\Alph{calBbCounter}}
    \expandafter\defcal\letter
		\expandafter\defbb\letter
}
\newcommand{\ie}{{\it i.e.}}

\newcommand{\nnR}{{\bR_{\geq 0}}}
\newcommand{\InDegree}{{d_{\mathrm{in}}}}
\newcommand{\OutDegree}{{d_{\mathrm{out}}}}

\DeclareMathOperator*{\argmax}{arg\,max}

\begin{document}

\title{Submodularity on Hypergraphs: From Sets to Sequences}

\author{\name Marko Mitrovic \email marko.mitrovic@yale.edu \\
       Yale University\\
       \AND
       \name Moran Feldman \email moranfe@openu.ac.il \\
       Open University of Israel\\
       \AND
       \name Andreas Krause \email krausea@ethz.ch \\
       ETH Zurich\\
       \AND
       \name Amin Karbasi \email amin.karbasi@yale.edu\\
       Yale University}
\maketitle


\begin{abstract}

In a nutshell, submodular functions encode an intuitive notion of diminishing returns. As a result, submodularity appears in many important machine learning tasks such as feature selection and data summarization. Although there has been a large volume of work devoted to the study of submodular functions in recent years, the vast majority of this work has been focused on algorithms that output sets, not sequences. However, in many settings, the order in which we output items can be just as important as the items themselves. 

To extend the notion of submodularity to sequences, we use a directed graph on the items where the edges encode the additional value of selecting items in a particular order. Existing theory is limited to the case where this underlying graph is a directed acyclic graph. In this paper, we introduce two new algorithms that provably give constant factor approximations for general graphs and hypergraphs having bounded in or out degrees. Furthermore, we show the utility of our new algorithms for real-world applications in movie recommendation, online link prediction, and the design of course sequences for MOOCs.
\end{abstract}
\section{Introduction}

\subsection{Preliminaries and Related Work}

Intuitively, submodularity describes the set of functions that exhibit diminishing returns. Mathematically, a set function $f : 2^{V} \to \mathbb{R}$ is \textbf{submodular} if, for every two sets $A \subseteq B \subseteq V$ and element $v \in V \setminus B$, we have $f(A \cup \{v\}) - f(A) \ge f(B \cup \{v\}) - f(B)$. That is, the marginal contribution of any element $v$ to the value of $f(A)$ diminishes as the set $A$ grows. 

As such, submodularity commonly appears in a~wide variety of fields including machine learning, combinatorial optimization, economics, and beyond. Sample applications include variable selection~\citep{krause05near}, data summarization~\citep{mirzasoleiman16distributed,lin2011class,kirchhoff2014submodularity}, recommender systems~\citep{GabillonKWEM2013}, crowd teaching~\citep{singla2014near}, neural network interpretability~\citep{elenberg17}, network monitoring~\citep{gomez10}, and influence maximization in social networks~\citep{kempe03}.

A submodular function $f$ is said to be \textbf{monotone} if $f(A) \leq f(B)$ for every two sets $A \subseteq B \subseteq V$. That is, adding items to a set cannot decrease its value. A seminal result in submodularity states that if our utility function $f$ is monotone submodular (and non-negative), then the classical greedy algorithm maximizes $f$ subject to a cardinality constraint up to an approximation ratio of $1 - 1/e$~\citep{nemhauser78}.  Since then, the study of submodular functions has been extended to a broad variety of different settings, including non-monotone submodularity~\citep{feige07maximizing,buchbinder2014submodular}, adaptive submodularity~\citep{golovin11}, weak submodularity~\citep{das2011submodular}, and continuous submodularity~\citep{wolsey82,bach2015}, just to name a few. 

Despite the above, the vast majority of existing results are limited to the scenario where we wish to output sets, not sequences. \cite{alaei10} and \cite{zhang16} consider functions they call string- or sequence-submodular, but it is in a different context. \cite{li17} look at a combination of submodularity and hypergraphs, but it is specifically within the context of hypergraph clustering. In this paper, we use a directed graph on the items where the edges encode the additional value of selecting items in a particular order. The only known theoretical result for this setting is limited to the case where the underlying graph is a directed acyclic graph~\citep{seq17}. Considering sequences instead of sets causes an exponential increase in the size of the search space, but it allows for much more expressive models. 

For example, consider the problem of recommending movies to a user. A recommendation system could determine that the user might be interested in The Lord of the Rings franchise. However, if the model does not consider the order of the movies it recommends, the user may watch The Return of the King first and The Fellowship of the Ring last, which is likely to make the user totally unsatisfied with an otherwise excellent recommendation. With this example as motivation, the next section gives a more detailed description of the problem we consider.

\subsection{Problem Description}

\citet{seq17} was the first to consider this particular submodular sequence setting and we will closely follow their setup of the problem. Recall that the goal is to select a \textbf{sequence} of items that will maximize some given objective function. To generalize the problem description, we will refer to items as vertices from now on.

Let $V = \{ v_1, v_2, \hdots, v_n \}$ be the set of $n$ vertices (items) we can pick from. A set of edges $E$ encodes the fact that there is additional value in picking certain vertices in a certain order. More specifically, an edge $e_{ij} = (v_i, v_j)$ encodes the fact that there is additional utility in selecting $v_j$ after $v_i$ has already been chosen. Self-loops (i.e., edges that begin and end at the same vertex) encode the fact that there is some individual utility in selecting a vertex.

In general, our input consists of a directed graph $G = (V, E)$, a non-negative monotone submodular set function $h \colon 2^E \to \nnR$, and a parameter $k$. The objective is to output a non-repeating sequence $\sigma$ of $k$ unique nodes that maximizes the objective function: 
\[
	f(\sigma)
	=
	h\big( E(\sigma) \big)
	\enspace,
\]
where
\[
E(\sigma) = \big\{ (\sigma_i, \sigma_j) \mid (\sigma_i, \sigma_j) \in E, i \leq j \big\}
\enspace.
\]

We say that $E(\sigma)$ is the set of edges induced by the sequence $\sigma$. It is important to note that the function $h$ is a submodular set function over the edges, not over the vertices. Furthermore, the objective function $f$ is neither a set function, nor is it necessarily submodular on the vertices.


\begin{figure}[h]
\vspace{.1in}
\begin{center}
\includegraphics[scale = 0.22]{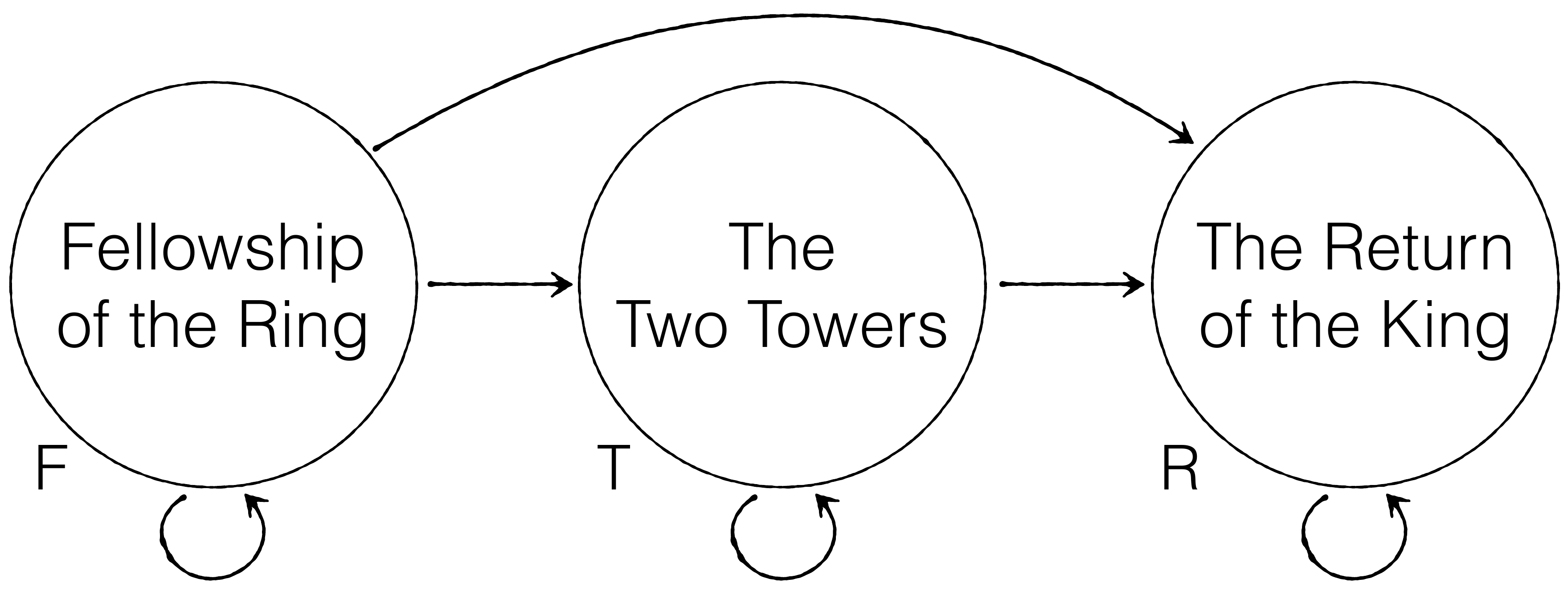}
\end{center}
\caption{Graph for The Lord of the Rings franchise. The self-loops encode the fact that each movie has some individual value. The edges encode the fact that there is additional utility in watching the movies in the correct order. Notice that the utility of watching The Return of the King after having already seen both The Fellowship of the Ring and The Two Towers is higher than the utility of watching The Return of the King after having seen just one of the two. }
\label{fig:lotr}
\end{figure}

For example, consider the graph in Figure~\ref{fig:lotr}, and let $h\big(E(\sigma) \big) = |E(\sigma)|$. That is, the value of a sequence is simply the number of edges induced by that sequence. Consider the sequence $\sigma_A = (F)$ where the user has watched only The Fellowship of the Ring, the sequence $\sigma_B = (T)$ where the user watched only The Two Towers, and the sequence $\sigma_C = (F,T)$ where the user watched The Fellowship of the Ring and then The Two Towers:\vspace{0.075in}

$f(\sigma_A) = f(F) = h\big( (F,F) \big) = 1 \enspace.$ \\
$f(\sigma_B) = f(T) = h\big( (T,T) \big) = 1 \enspace.$ \\
$f(\sigma_C) = f(F,T) = h\big( (F,F),(F,T),(T,T) \big) = 3\enspace.$ \vspace{0.075in}

This example shows that although the marginal gain of the edges is non-increasing in the context of a growing set of edges (i.e., the function $h$ is submodular on the edges), it is clear that the function $f$ is \emph{not} submodular on the vertices. In particular, the marginal gain of The Two Towers is larger once the user has already viewed The Fellowship of the Ring.

Furthermore, just to fully clarify the concept of edges being induced by a sequence, consider the sequence $\sigma_D = (T,F)$ where the user watched The Two Towers and then The Fellowship of the Ring. \vspace{0.075in}

$f(\sigma_D) = f(T,F) = h\big( (T,T),(F,F) \big) = 2 \enspace.$ 
\vspace{0.075in}

Notice that although sequences $\sigma_C$ and $\sigma_D$ contain the same movies, the order of $\sigma_D$ means that the edge $(F,T)$ is not induced, and thus, the value of the sequence is lower.

\subsection{Our Contributions}
\vspace{-0.1in}

Throughout this paper we use the notation $\Delta = \min\{ d_{\text{in}}, d_{\text{out}}\}$, where $d_{\text{in}} = \max_{v \in V}\InDegree(v)$ and $d_{\text{out}} = \max_{v \in V}\OutDegree(v)$. The previous work on our problem, due to~\citet{seq17}, presented an algorithm (OMegA) enjoying a $(1 - e^{-\frac{1}{2\Delta}})$-approximation guarantee when the underlying graph $G$ is a directed acyclic graph (except for self-loops).  

In this paper, we present two new algorithms: Sequence-Greedy and Hyper Sequence-Greedy, which also provably achieve constant factor approximations (when $\Delta$ is constant), but their guarantees hold for general graphs and hypergraphs, respectively. Although the example given in Figure \ref{fig:lotr} is indeed a directed acyclic graph, many real-world problems require a general graph or hypergraph.  

We showcase the utility of our algorithms on real world applications in movie recommendation, online link prediction, and the design of course sequences for massive open online courses (MOOCs). Furthermore, we show that even when the underlying graph is a directed acyclic graph, our general graph algorithm performs comparably well. Our experiments also demonstrate the power of being able to utilize hypergraphs and hyperedges.

 \paragraph{Paper Organization.} Our theoretical results are formally presented in Section~\ref{theory}, and their proofs can be found in Section~\ref{proofs}. Details about the real-world applications we consider and the experimental results we obtain for these applications appear in Section~\ref{apps}.

\newpage

\section{Theoretical Results} \label{theory}
\subsection{General Graphs}\label{general}

In this section, we present our first algorithm, Sequence-Greedy. Sequence-Greedy is essentially the same as the classical greedy algorithm, but instead of choosing the most valuable vertex at each step, it chooses the most valuable valid edge. 

More specifically, we start off with an empty sequence $\sigma$. At each step, we define $\cE$ to be the set of all edges whose end point is not already in $\sigma$. We then greedily select the edge $e_{ij} \in \cE$ with maximum marginal gain $h\big(e_{ij} \mid E(\sigma)\big)$, where
\[
h\big(e_{ij} \mid E(\sigma)\big) = h\big(E(\sigma) \cup e_{ij}  \big) - h\big(E(\sigma)\big)
\enspace.
\]
Recall that $e_{ij} = (v_i,v_j)$. That is, $v_i$ is the start point of $e_{ij}$ and $v_j$ is the endpoint. If $e_{ij}$ is a self-loop, then $j = i$ and we append the single vertex $v_j$ to $\sigma$. Similarly, if $j \neq i$, but $v_i$ is already in $\sigma$, then we still only append $v_j$. Finally, if $e_{ij}$ has two distinct vertices and neither of them is already in the sequence, we append $v_i$ and then $v_j$ to $\sigma$. This description is summarized in pseudo-code in Algorithm~\ref{alg:greedy_path}. 


\begin{algorithm}[h]
 \SetKwInOut{Input}{Input}
\caption{\textsf{Sequence-Greedy (Forward)}} \label{alg:greedy_path}
\DontPrintSemicolon
%

\Input{Directed graph $G = (V,E)$\\ Monotone submodular function $h: 2^E \rightarrow \mathbb{R}$\\ Cardinality parameter $k$}
Let $\sigma \gets ()$.\\
\While{$|\sigma| \leq k - 2$}
{
	$\cE = \{ e_{ij} \in E \mid v_j \notin \sigma \}$. \tcp*{$e_{ij} = (v_i,v_j)$} 
	\lIf{$\cE = \varnothing$}{Exit the loop.}
	$e_{ij} = \argmax_{e \in \cE} h(e \mid E(\sigma))$. \\
	\eIf{$v_j = v_i$ \textbf{or} $v_i \in \sigma$}{$\sigma = \sigma \oplus v_j$. 	 \tcp*{$\oplus$ means concatenate}} 
	{$\sigma = \sigma \oplus v_i \oplus v_j$.}
}

\Return{$\sigma$}.
\end{algorithm}

\begin{theorem} \label{thm:graphs}
The approximation ratio of Algorithm~\ref{alg:greedy_path} is at least $\frac{1 - e^{-(1 - \frac{1}{k})}}{2d_{\text{in}} + 1}$.
\end{theorem}

Notice that the approximation guarantee of Algorithm~\ref{alg:greedy_path} depends on the maximum in-degree $d_{\text{in}}$. Intuitively, this is because Algorithm~\ref{alg:greedy_path} builds $\sigma$ by appending vertices to the end of the sequence. This means that each vertex we add to $\sigma$ decreases the size of $\cE$ by at most $d_{\text{in}}$.

However, one can easily modify Algorithm~\ref{alg:greedy_path} to build $\sigma$ backwards by \emph{prepending} vertices to the start of the sequence at each step. More specifically, we redefine $\cE$ to be the set of all edges whose \emph{start point} is not already in $\sigma$. Again we greedily select the edge $e_{ij} \in \cE$ that maximizes $h\big(e_{ij} \mid E(\sigma)\big)$. Now, if $e_{ij}$ is a self-loop or $v_j$ is already in $\sigma$, we prepend the single vertex $v_i$ to the start of $\sigma$. Otherwise, if $e_{ij}$ has two distinct vertices and neither of them is already in the sequence, we prepend $v_j$ to $\sigma$ first, and then prepend $v_i$ (thus, maintaining the order). This description is summarized in pseudo-code in Algorithm~\ref{alg:greedy_path2} with the main differences noted as comments.

\begin{algorithm}[h]
 \SetKwInOut{Input}{Input}
\caption{\textsf{Sequence-Greedy (Backward)}} \label{alg:greedy_path2}
\DontPrintSemicolon
\Input{Directed graph $G = (V,E)$\\ Monotone submodular function $h: 2^E \rightarrow \mathbb{R}$\\ Cardinality parameter $k$}
Let $\sigma \gets ()$.\\
\While{$|\sigma| \leq k - 2$}
{
	$\cE = \{ e_{ij} \in E \mid v_i \notin \sigma \}$. \tcp*{different set $\cE$} 
	\lIf{$\cE = \varnothing$}{Exit the loop.}
	$e_{ij} = \argmax_{e \in \cE} h(e \mid E(\sigma))$.  \\
	\eIf{$v_i = v_j$ \textbf{or} $v_j \in \sigma$}{$\sigma = v_i \oplus \sigma$. } 
	{$\sigma = v_i \oplus v_j \oplus \sigma$.} \tcp*{vertices appended to beginning of $\sigma$} 

}

\Return{$\sigma$}.
\end{algorithm}

Algorithm~\ref{alg:greedy_path2} gives the same approximation ratio as Algorithm~\ref{alg:greedy_path}, but with a dependence on $d_{\text{out}}$ instead of $d_{\text{in}}$. Thus, if we run both the forwards and backwards version of Sequence-Greedy and take the maximum, we get an approximation ratio that depends on $\Delta = \min\{d_{\text{in}}, d_{\text{out}}\}$. Furthermore, notice that the approximation ratio improves as $k$ increases. Therefore, we can summarize the approximation ratio of Sequence-Greedy as follows.

\begin{theorem}
As $k \rightarrow \infty$, the approximation ratio of Sequence-Greedy approaches $\frac{1 - \frac{1}{e}}{2\Delta + 1}$.
\end{theorem}

\newpage
This is comparable to the $(1 - e^{-\frac{1}{2\Delta}})$-approximation guarantee that is achieved by the existing algorithm OMegA, except that our guarantee is valid on general graphs, not just directed acyclic graphs.

In addition to this provable approximation ratio, Sequence-Greedy has the strong advantage of being computationally efficient. Both finding $\cE$ and identifying the most valuable edge in $\cE$ can be done in $O(m)$ time, where $m = |E|$. Thus, Sequence-Greedy runs in $O(km)$ time. This is faster than OMegA, which runs in $O(m\Delta k^2 \log k)$.

\vspace{-0.0in}
\subsection{Extension to Hypergraphs}
\vspace{-0.0in}

Extending our results to hypergraphs allows us to encode increasingly sophisticated models. For example, looking back on Figure \ref{fig:lotr}, we see that the~value of watching all three movies is just the sum of the pairwise additional values. However, hyperedges allow us to encode the fact that there is even further utility in watching the entire franchise in order. 

From this point on, we replace the directed graph $G$ with a directed hypergraph $H = (V,E)$. Each edge $e \in E$ of this directed hypergraph is a non-empty non-repeating sequence of vertices from $V$. Let $V(e)$ be the set of vertices found in the hyperedge $e$. We assume that the intersection of a sequence and a set maintains the order of the sequence, which allows us to redefine $E(\sigma)$ as
\[
	E(\sigma) = \{e \in E \mid \sigma \cap V(e) = e\} \enspace.
\]
Informally, $E(\sigma)$ contains an edge $e \in E$ if and only if all the vertices of $e$ appear in $\sigma$ in the proper order.

We also need to explain how the concept of in-degrees and out-degrees extends to hypergraphs. Self-loops contribute 1 to both the in-degree and the out-degree of that vertex. For all other edges $e \in E$ such that $v \in V(e)$, they will contribute 1 to $\InDegree(v)$ if $v$ is not the first vertex of $e$, and 1 to $\OutDegree(v)$ if $v$ is not the last vertex of $e$. Finally, we define $r$ as the maximum size of any edge in $E$. More formally, $r = \max_{e \in E} |e|$.

Aside from the above redefinition of $E(\sigma)$, there is no need to make other changes in the definition of the objective function $f$. Specifically, it is still defined as $f(\sigma) = h(E(\sigma))$, where $h\colon 2^E \to \nnR$ is a non-negative monotone submodular function. 


Our algorithm for hypergraphs, Hyper Sequence-Greedy, is an extension of the original Sequence-Greedy. Again, we start off with an empty sequence $\sigma$. This time, at each step we define $\cE$ to be the set of all hyperedges $e \in E$ such that $\sigma \cap V(e)$ is a prefix of $e$. The idea is that we can only select a hyperedge $e$ if the vertices of $e$ already in our sequence $\sigma$ form a prefix of $e$, and they appear in $\sigma$ in the right order. We then select the hyperedge $e^{*} \in \cE$ that has the maximum marginal gain, and append the vertices of $e^*$ (that are not already in our sequence) to $\sigma$ without changing their order. This description is summarized in pseudo-code in Algorithm~\ref{alg:greedy_path_hyper}.


\begin{algorithm}[ht]
 \SetKwInOut{Input}{Input}
\caption{\textsf{Hyper Sequence-Greedy (Forward)}} \label{alg:greedy_path_hyper}
\DontPrintSemicolon
%

 \Input{Directed hypergraph $H = (V,E)$\\ Monotone submodular function $h$\\ Cardinality parameter $k$}
Let $\sigma \gets ()$.\\
\While{$|\sigma| \leq k - r$}
{
	Let $\cE = \{e \in E \mid \sigma \cap V(e) \text{ is a prefix of } e\}$.\\
	\lIf{$\cE = \varnothing$}{Exit the loop.}
	$e^* = \argmax_{e \in \cE} h(e \mid E(\sigma))$.  \\
	\For{every $v \in e^*$ in order}{\lIf{$v \notin \sigma$}{$\sigma = \sigma \oplus v$. }}
	
}

\Return{$\sigma$}.
\end{algorithm}

\begin{theorem} \label{thm:hyper}
The approximation ratio of Algorithm~\ref{alg:greedy_path_hyper} is at least $\frac{1 - e^{-(1 - \frac{r}{k})}}{r\InDegree + 1}$.
\end{theorem}

As with Sequence-Greedy, we can also run Hyper Sequence-Greedy backwards and take the maximum of the two results. In the backwards version, we \emph{prepend} the vertices to the start of the sequence and we can only select a hyperedge $e$ if $V(e) \cap \sigma$ is a suffix of $e$. Once more, this improves the approximation ratio in the sense that the dependence on $\InDegree$ is replaced with a dependence on $\Delta = \min\{\InDegree,\OutDegree\}$. Additionally notice that, as before, our approximation ratio improves as $k$ increases. Thus, we can summarize the performance guarantee of Hyper Sequence-Greedy as follows.

\begin{theorem}
As $k \rightarrow \infty$, the approximation ratio of Hyper Sequence-Greedy approaches $\frac{1 - \frac{1}{e}}{r\Delta + 1}$.
\end{theorem}

\noindent \textbf{Remarks:} One can observe that this hypergraph setting is a generalization of the previous directed graph setting. Specifically, Sequence-Greedy and the associated theory is a special case of Hyper Sequence-Greedy for $r=2$. Furthermore, if $r=1$ (i.e., our graph has only self-loops) then Hyper Sequence-Greedy is the same as the classical greedy algorithm.

We also note that while Algorithm~\ref{alg:greedy_path_hyper} may select fewer than $k$ vertices, the theoretical guarantees still hold. Furthermore, since we assume that $h$ is monotone, we can safely select $k$ vertices in practice every time. One simple heuristic for extending $\sigma$ to $k$ vertices is to only consider hyperedges with at most $k - |\sigma|$ vertices.

\section{Proofs} \label{proofs}

\subsection{Proof of Theorem~\ref{thm:graphs}} \label{app:graphs}

In this section we prove Theorem~\ref{thm:graphs}, however, before we get into the proof, let us first recall the theorem itself.

\begin{reptheorem}{thm:graphs}
The approximation ratio of Algorithm~\ref{alg:greedy_path} is at least $\frac{1 - e^{-(1 - \frac{1}{k})}}{2d_{\text{in}} + 1}$.
\end{reptheorem}

We begin the proof of the theorem by defining some additional notation. First, let $\ell$ be the number of iterations completed by the main loop of Algorithm~\ref{alg:greedy_path}, \ie, the number of iterations in which $\sigma$ is updated. Then, for every $0 \leq s \leq \ell$, let $\sigma_s$ be the value of $\sigma$ after $s$ iterations of this loop have been performed. In other words, $\sigma_0$ is the initial value of $\sigma$ when we first get to the loop, $\sigma_1$ is the value of $\sigma$ at the end of the first iteration of the loop, and so on. Note that $\sigma_\ell$ is the output of Algorithm~\ref{alg:greedy_path}. Additionally, we also denote by $e_s$ and $\cE_s$, for every $1 \leq s \leq \ell$, the values assigned to the variables $e_{ij}$ and $\cE$, respectively, at iteration number $s$ of the above loop. Finally, we also define $\ell'$ as the real number of iterations performed by the above loop. Notice that $\ell' = \ell$ unless the algorithm exits the loop because $\cE = \varnothing$, in which case $\ell' = \ell + 1$ and we define $\cE_{\ell'} = \varnothing$.

\begin{observation} \label{obs:gain}
For every $1 \leq s \leq \ell$, $f(\sigma_s) - f(\sigma_{s - 1}) \geq h(e_s \mid E(\sigma_{s - 1}))$.
\end{observation}
\begin{proof}
Notice that the way $\sigma$ is updated in each iteration of Algorithm~\ref{alg:greedy_path} guarantees that $e_s \in E(\sigma_s) \setminus E(\sigma_{s - 1})$. Moreover, since $\sigma_{s - 1}$ is a prefix of $\sigma_s$, we also get $E(\sigma_{s - 1}) \subseteq E(\sigma_s)$. Thus,
\[
	f(\sigma_s) - f(\sigma_{s - 1})
	=
	h(E(\sigma_s)) - h(E(\sigma_{s - 1}))
	\geq
	h(E(\sigma_{s - 1}) + e_s) - h(E(\sigma_{s - 1}))
	=
	h(e_s \mid E(\sigma_{s - 1}))
	\enspace,
\]
where the inequality follows from the monotonicity of $h$.
\end{proof}

Let $\sigma^*$ denote an arbitrary (but fixed) optimal sequence. We now need to prove a few properties of $\sigma^*$.

\begin{observation} \label{obs:size}
$|E(\sigma^*)| \leq \InDegree k$.
\end{observation}
\begin{proof}
Observe that $\sigma^*$ contains at most $k$ vertices because it is feasible. This means that there can be at most $\InDegree k$ arcs that end in a vertex of $\sigma^*$, which implies the observation since every arc of $E(\sigma^*)$ must end at a vertex of $\sigma^*$.
\end{proof}

The next lemma studies the change in the value of $(E(\sigma^*) \cap \cE_s) \cup E(\sigma_{s - 1})$ as a function of $s$.

\begin{lemma} \label{lem:loss}
For every $1 \leq s < \ell'$, $h((E(\sigma^*) \cap \cE_{s + 1}) \cup E(\sigma_s)) \geq h((E(\sigma^*) \cap \cE_s) \cup E(\sigma_{s-1})) - 2\InDegree \cdot h(e_s \mid \sigma_{s - 1})$.
\end{lemma}
\begin{proof}
Recall that, for every $1 \leq s \leq \ell'$, $\cE_s$ contains the arcs of $E$ whose end point is not in $\sigma_{s - 1}$. This definition implies that $\cE_{s + 1} \subseteq \cE_s$ because $\sigma_{s - 1}$ is a prefix of $\sigma_s$. In contrast, since $\sigma_s$ contains at most two vertices that do not appear in $\sigma_{s - 1}$ and each one of these vertices can be the end point of at most $\InDegree$ arcs, we also get $|\cE_s \setminus \cE_{s + 1}| \leq 2\InDegree$.

Using these observations and the submodularity of $h$, we can now derive the following inequality.
\begin{align*}
	h(E(\sigma^*) \cap \cE_s \mid E(\sigma_{s-1})) -{}&h(E(\sigma^*) \cap \cE_{s + 1} \mid E(\sigma_{s-1}))
	\leq
	\sum_{e \in E(\sigma^*) \cap (\cE_s \setminus \cE_{s + 1})} \mspace{-36mu} h(e \mid E(\sigma_{s-1}))\\
	\leq{} &
	\sum_{e \in E(\sigma^*) \cap (\cE_s \setminus \cE_{s + 1})} \mspace{-36mu} h(e_s \mid E(\sigma_{s-1}))
	=
	|E(\sigma^*) \cap (\cE_s \setminus \cE_{s + 1})| \cdot h(e_s \mid E(\sigma_{s-1}))\\
	\leq{} &
	|\cE_s \setminus \cE_{s + 1}| \cdot h(e_s \mid E(\sigma_{s-1}))
	\leq
	2\InDegree \cdot h(e_s \mid E(\sigma_{s-1}))
	\enspace,
\end{align*}
where the second inequality follows from the definition of $e_s$ which guarantees that it maximizes $h(e_s \mid E(\sigma_{s-1}))$ among all the edges of $\cE_s$.

It now remains to observe that
\begin{align*}
	h((E(\sigma^*) \cap \cE_s) \cup{}& E(\sigma_{s-1})) - h((E(\sigma^*) \cap \cE_{s + 1}) \cup E(\sigma_s))\\
	\leq{} &
	h((E(\sigma^*) \cap \cE_s) \cup E(\sigma_{s-1})) - h((E(\sigma^*) \cap \cE_{s + 1}) \cup E(\sigma_{s - 1}))\\
	={} &
	h(E(\sigma^*) \cap \cE_s \mid E(\sigma_{s-1})) - h(E(\sigma^*) \cap \cE_{s + 1} \mid E(\sigma_{s - 1}))
	\leq
	2\InDegree \cdot h(e_s \mid E(\sigma_{s-1}))
	\enspace,
\end{align*}
where the first inequality follows from the monotonicity of $h$ since the fact that $\sigma_{s - 1}$ is a prefix of $\sigma_s$ implies $E(\sigma_{s - 1}) \subseteq E(\sigma_s)$.
\end{proof}

We are now ready to combine all the above claims into a single lemma.

\begin{lemma} \label{lem:iteration_guarantee}
For every $1 \leq s \leq \ell$, the following two inequalities hold:
\begin{itemize}
	\item $h((E(\sigma^*) \cap \cE_s) \cup E(\sigma_{s - 1})) \geq f(\sigma^*) - 2\InDegree \cdot f(\sigma_{s - 1})$.
	\item $f(\sigma_s) - f(\sigma_{s - 1}) \geq \frac{f(\sigma^*) - 2\InDegree \cdot [f(\sigma_{s - 1}) - f(\sigma_0)] - f(\sigma_{s - 1})}{\InDegree k} \geq \frac{f(\sigma^*) - (2\InDegree + 1) \cdot f(\sigma_{s - 1})}{\InDegree k}$.
\end{itemize}
Moreover, the first inequality holds also for $s = \ell'$.
\end{lemma}
\begin{proof}
Lemma~\ref{lem:loss} shows that, for every $1 \leq t < \ell'$, we have
\[
	h((E(\sigma^*) \cap \cE_{t + 1}) \cup E(\sigma_t)) \geq h((E(\sigma^*) \cap \cE_t) \cup E(\sigma_{t-1})) - 2\InDegree \cdot h(e_t \mid \sigma_{t - 1})
	\enspace.
\]
Adding up this inequality for $1 \leq t < s$ gives us
\begin{align*}
	h((E(\sigma^*) \cap \cE_s) \cup E(\sigma_{s - 1}))
	\geq{} &
	h((E(\sigma^*) \cap \cE_1) \cup E(\sigma_0)) - 2\InDegree \cdot \sum_{t = 1}^{s - 1} h(e_t \mid \sigma_{t - 1})\\
	={}&
	f(\sigma^*) - 2\InDegree \cdot \sum_{t = 1}^{s - 1} h(e_t \mid \sigma_{t - 1})
	\geq
	f(\sigma^*) - 2\InDegree \cdot \sum_{t = 1}^{s - 1} [f(\sigma_t) - f(\sigma_{t - 1})]\\
	={} &
	f(\sigma^*) - 2\InDegree \cdot [f(\sigma_{s - 1}) - f(\sigma_0)]
	\geq
	f(\sigma^*) - 2\InDegree \cdot f(\sigma_{s - 1})
	\enspace.
\end{align*}
The first equality follows since the fact that $\sigma_0$ is empty implies $E(\sigma_0) = \varnothing$ and $E(\sigma^*) \cap \cE_1 = E(\sigma^*)$. 
Additionally, the second inequality follows from Observation~\ref{obs:gain}, and the last inequality follows from the non-negativity of $f$. This proves that the first inequality of the lemma holds for every $1 \leq s \leq \ell'$. In the rest of the proof we aim to prove the second inequality, and thus, assume $1 \leq s \leq \ell$.

Recall now that $e_s$ is the edge of $\cE_s$ maximizing $h(e_s \mid E(\sigma_{s - 1}))$ and that the size of $E(\sigma^*) \cap \cE_s \subseteq E(\sigma^*)$ is at most $\InDegree k$ by Observation~\ref{obs:size}. Thus, by the submodularity of $h$,
\begin{align*}
	h(e_s \mid E(\sigma_{s - 1}))
	\geq{} &
	\frac{\sum_{e \in E(\sigma^*) \cap \cE_s} h(e \mid E(\sigma_{s - 1}))}{|E(\sigma^*) \cap \cE_s|}
	\geq
	\frac{h(E(\sigma^*) \cap \cE_s \mid E(\sigma_{s - 1}))}{\InDegree k}\\
	={} &
	\frac{h((E(\sigma^*) \cap \cE_s) \cup E(\sigma_{s - 1})) - h(E(\sigma_{s - 1}))}{\InDegree k}\\
	\geq{} &
	\frac{\{f(\sigma^*) - 2\InDegree \cdot [f(\sigma_{s - 1}) - f(\sigma_0)]\} - f(\sigma_{s - 1})}{\InDegree k}
	\geq
	\frac{f(\sigma^*) - (2\InDegree + 1) \cdot f(\sigma_{s - 1})}{\InDegree k}
	\enspace.
\end{align*}
The lemma follows from this inequality since $h(e_s \mid E(\sigma_{s - 1}))$ lower bounds $f(\sigma_s) - f(\sigma_{s - 1})$ by Observation~\ref{obs:gain}.
\end{proof}

\begin{corollary} \label{cor:final_raw}
If $2\InDegree + 1 < \InDegree k$, then $f(\sigma_{\ell}) \geq \frac{f(\sigma^*)}{2\InDegree + 1} + \frac{[1 - (2\InDegree + 1) / (\InDegree k)]^\ell}{2\InDegree + 1} \cdot [(2\InDegree + 1)f(\sigma_0) - f(\sigma^*)]$.
\end{corollary}
\begin{proof}
To prove the corollary, we prove by induction the stronger claim that, for every $0 \leq s \leq \ell$,
\[
	f(\sigma_s)
	\geq
	\frac{f(\sigma^*)}{2\InDegree + 1} + \frac{[1 - (2\InDegree + 1) / (\InDegree k)]^s}{2\InDegree + 1} \cdot [(2\InDegree + 1)f(\sigma_0) - f(\sigma^*)]
	\enspace.
\]
For $s = 0$ this inequality is true since
\begin{align*}
	f(\sigma_0)
	={} &
	\frac{f(\sigma^*)}{2\InDegree + 1} + \frac{1}{2\InDegree + 1} \cdot [(2\InDegree + 1)f(\sigma_0) - f(\sigma^*)]\\
	={} &
	\frac{f(\sigma^*)}{2\InDegree + 1} + \frac{[1 - (2\InDegree + 1) / (\InDegree k)]^0}{2\InDegree + 1} \cdot [(2\InDegree + 1)f(\sigma_0) - f(\sigma^*)]
	\enspace.
\end{align*}
Assume now that the claim holds for $s - 1 \geq 0$, and let us prove it for $s$. By Lemma~\ref{lem:iteration_guarantee},
\[
	f(\sigma_s)
	\geq
	f(\sigma_{s - 1}) + \frac{f(\sigma^*) - (2\InDegree + 1) \cdot f(\sigma_{s - 1})}{\InDegree k}
	=
	\left(1 - \frac{2\InDegree + 1}{\InDegree k}\right) \cdot f(\sigma_{s - 1}) + \frac{f(\sigma^*)}{\InDegree k}
	\enspace.
\]
Plugging in the induction hypothesis, we get
\begin{align*}
	f(\sigma_s)
	\geq{} &
	\left(1 - \frac{2\InDegree + 1}{\InDegree k}\right) \cdot \left\{\frac{f(\sigma^*)}{2\InDegree + 1} + \frac{[1 - (2\InDegree + 1) / (\InDegree k)]^{s - 1}}{2\InDegree + 1} \cdot [(2\InDegree + 1)f(\sigma_0) - f(\sigma^*)]\right\} \\&+ \frac{f(\sigma^*)}{\InDegree k}
	=
	\frac{f(\sigma^*)}{2\InDegree + 1} + \frac{[1 - (2\InDegree + 1) / (\InDegree k)]^s}{2\InDegree + 1} \cdot [(2\InDegree + 1)f(\sigma_0) - f(\sigma^*)]
	\enspace.
	\qedhere
\end{align*}
\end{proof}

We are now ready to prove Theorem~\ref{thm:graphs}.

\begin{proof}[Proof of Theorem~\ref{thm:graphs}]
First, we need to consider the case that Algorithm~\ref{alg:greedy_path} terminates because the set $\cE$ becomes empty. In this case $\cE_{\ell'} = \varnothing$, which implies
\[
	h((E(\sigma^*) \cap \cE_{\ell'}) \cup E(\sigma_{\ell' - 1}))
	=
	h(E(\sigma_\ell))
	=
	f(\sigma_\ell).
\]
Using Lemma~\ref{lem:iteration_guarantee} for $s = \ell'$, this observation implies
\[
	f(\sigma_\ell)
	\geq
	f(\sigma^*) - 2\InDegree \cdot f(\sigma_{\ell})
	\Rightarrow
	f(\sigma_{\ell})
	\geq
	\frac{f(\sigma^*)}{2\InDegree + 1}
	\enspace,
\]
which proves the theorem. Thus, in the rest of the proof we may assume that Algorithm~\ref{alg:greedy_path} terminates because $\sigma$ reaches a size larger than $k - 2$.

Consider now the case that $2\InDegree + 1 \geq \InDegree k$. In this case
\begin{align*}
	f(\sigma_{\ell})
	\geq{} &
	f(\sigma_1)
	=
	f(\sigma_0) + [f(\sigma_1) - f(\sigma_0)]
	\geq
	f(\sigma_0) + \frac{f(\sigma^*) - f(\sigma_0)}{\InDegree k}\\
	\geq{} &
	\frac{f(\sigma^*)}{\InDegree k}
	\geq
	\frac{f(\sigma^*)}{2\InDegree + 1}
	\geq
	\frac{1 - e^{-(1 - \frac{1}{k})}}{2\InDegree + 1} \cdot f(\sigma^*)
	\enspace,
\end{align*}
where the first inequality holds since $\sigma_1$ is a prefix of $\sigma_\ell$ and the second inequality follows from Lemma~\ref{lem:iteration_guarantee}. Thus, it remains to prove the theorem in the more interesting case of $2\InDegree + 1 < \InDegree k$.


%
Observe that
\[
	[1 - (2\InDegree + 1) / (\InDegree k)]^\ell
	\leq
	e^{-(2 + 1 / \InDegree) \cdot (\ell / k)}
	\enspace.
\]
Plugging this inequality into Corollary~\ref{cor:final_raw} gives
\begin{align*}
	f(\sigma_{\ell})
	\geq{} &
	\frac{f(\sigma^*)}{2\InDegree + 1} + \frac{[1 - (2\InDegree + 1) / (\InDegree k)]^\ell}{2\InDegree + 1} \cdot [(2\InDegree + 1)f(\sigma_0) - f(\sigma^*)]\\
	\geq{} &
	\frac{1 - [1 - (2\InDegree + 1) / (\InDegree k)]^\ell}{2\InDegree + 1} \cdot f(\sigma^*)
	\geq
	\frac{1 - e^{-(2 + 1 / \InDegree) \cdot (\ell / 	k)}}{2\InDegree + 1} \cdot f(\sigma^*)
	\enspace.
\end{align*}
At this point we need a lower bound on $\ell$. One can note that $|\sigma|$ starts as $0$, increases by at most $2$ in each iteration of the loop of Algorithm~\ref{alg:greedy_path} and ends up with a value of at least $k - 1$ by our assumption. Thus, the number $\ell$ of iterations must be at least $(k - 1)/2$. Plugging this observation into the previous inequality gives
\[
	f(\sigma_{\ell})
	\geq
	\frac{1 - e^{-(2 + 1 / \InDegree) \cdot (1 - 1 / k)/2}}{2\InDegree + 1} \cdot f(\sigma^*)
	\geq
	\frac{1 - e^{-(1 - 1 / k)}}{2\InDegree + 1} \cdot f(\sigma^*)
	\enspace.
	\qedhere
\]
\end{proof}


\subsection{Proof of Theorem~\ref{thm:hyper}}

In this section we prove Theorem~\ref{thm:hyper}, however, before we get into the proof, let us first recall the theorem itself.

\begin{reptheorem}{thm:hyper}
The approximation ratio of Algorithm~\ref{alg:greedy_path_hyper} is at least $\frac{1 - e^{-(1 - \frac{r}{k})}}{r\InDegree + 1}$.
\end{reptheorem}

In the proof of this theorem we use the same notation that we used in Section~\ref{app:graphs} for analyzing Algorithm~\ref{alg:greedy_path}. One can observe that the proofs of Observation~\ref{obs:gain} and Observation~\ref{obs:size} are unaffected by the differences between Algorithm~\ref{alg:greedy_path} and Algorithm~\ref{alg:greedy_path_hyper}, and thus, these two observations can also be used for towards the proof of Theorem~\ref{thm:hyper}.

The next lemma is analogous to Lemma~\ref{lem:loss}.

\begin{lemma} \label{lem:loss_hyper}
For every $1 \leq s < \ell'$, $h((E(\sigma^*) \cap \cE_{s + 1}) \cup E(\sigma_s)) \geq h((E(\sigma^*) \cap \cE_s) \cup E(\sigma_{s-1})) - r\InDegree \cdot h(e_s \mid \sigma_{s - 1})$.
\end{lemma}
\begin{proof}
Observe that the definition of $\cE_s$ guarantees that $\cE_{s + 1} \subseteq \cE_s$, for every $1 \leq s < \ell'$, because $\sigma_{s - 1}$ is a prefix of $\sigma_s$.
In contrast, every vertex $u$ that appears in $\sigma_s$ but not in $\sigma_{s - 1}$ can be responsible for at most $\InDegree$ arcs of $\cE_s \setminus \cE_{s + 1}$ because $u$ can be responsible for excluding an arc from $\cE_{s + 1}$ only if $u$ is a non-first vertex of the arc. Since $\sigma_s$ contains at most $r$ vertices that do not appear in $\sigma_{s - 1}$, this implies $|\cE_s \setminus \cE_{s + 1}| \leq r\InDegree$.

Using these observations and the submodularity of $h$, we can now derive the following inequality.
\begin{align*}
	h(E(\sigma^*) \cap \cE_s \mid E(\sigma_{s-1})) -{}&h(E(\sigma^*) \cap \cE_{s + 1} \mid E(\sigma_{s-1}))
	\leq
	\sum_{e \in E(\sigma^*) \cap (\cE_s \setminus \cE_{i + 1})} \mspace{-36mu} h(e \mid E(\sigma_{s-1}))\\
	\leq{} &
	\sum_{e \in E(\sigma^*) \cap (\cE_s \setminus \cE_{i + 1})} \mspace{-36mu} h(e_s \mid E(\sigma_{s-1}))
	=
	|E(\sigma^*) \cap (\cE_s \setminus \cE_{s + 1})| \cdot h(e_s \mid E(\sigma_{s-1}))\\
	\leq{} &
	|\cE_s \setminus \cE_{s + 1}| \cdot h(e_s \mid E(\sigma_{s-1}))
	\leq
	r\InDegree \cdot h(e_s \mid E(\sigma_{s-1}))
	\enspace,
\end{align*}
where the second inequality follows from the definition of $e_s$ which guarantees that it maximizes $h(e_s \mid E(\sigma_{s-1}))$ among all the edges of $\cE_s$.

It now remains to observe that
\begin{align*}
	h((E(\sigma^*) \cap \cE_s) \cup{}& E(\sigma_{s-1})) - h((E(\sigma^*) \cap \cE_{s + 1}) \cup E(\sigma_s))\\
	\leq{} &
	h((E(\sigma^*) \cap \cE_s) \cup E(\sigma_{s-1})) - h((E(\sigma^*) \cap \cE_{s + 1}) \cup E(\sigma_{s - 1}))\\
	={} &
	h(E(\sigma^*) \cap \cE_s \mid E(\sigma_{s-1})) - h(E(\sigma^*) \cap \cE_{s + 1} \mid E(\sigma_{s - 1}))
	\leq
	r\InDegree \cdot h(e_s \mid E(\sigma_{s-1}))
	\enspace,
\end{align*}
where the first inequality follows from the monotonicity of $h$ since the fact that $\sigma_{s - 1}$ is a prefix of $\sigma_s$ implies $E(\sigma_{s - 1}) \subseteq E(\sigma_s)$.
\end{proof}

We are now ready to prove the following analog of Lemma~\ref{lem:iteration_guarantee}.

\begin{lemma} \label{lem:iteration_guarantee_hyper}
For every $1 \leq s \leq \ell$, the following two inequalities hold:
\begin{itemize}
	\item $h((E(\sigma^*) \cap \cE_s) \cup E(\sigma_{s - 1})) \geq f(\sigma^*) - r\InDegree \cdot f(\sigma_{s - 1})$.
	\item $f(\sigma_s) - f(\sigma_{s - 1}) \geq \frac{f(\sigma^*) - r\InDegree \cdot [f(\sigma_{s - 1}) - f(\sigma_0)] - f(\sigma_{s-1})}{\InDegree k} \geq \frac{f(\sigma^*) - (r\InDegree + 1) \cdot f(\sigma_{s - 1})}{\InDegree k}$.
\end{itemize}
Moreover, the first inequality holds also for $s = \ell'$.
\end{lemma}
\begin{proof}
Lemma~\ref{lem:loss_hyper} shows that, for every $1 \leq t < \ell'$, we have
\[
	h((E(\sigma^*) \cap \cE_{t + 1}) \cup E(\sigma_t)) \geq h((E(\sigma^*) \cap \cE_t) \cup E(\sigma_{t-1})) - r\InDegree \cdot h(e_t \mid \sigma_{t - 1})
	\enspace.
\]
Adding up this inequality for $1 \leq t < s$ gives us
\begin{align*}
	h((E(\sigma^*) \cap \cE_s) \cup E(\sigma_{s - 1}))
	\geq{} &
	h((E(\sigma^*) \cap \cE_1) \cup E(\sigma_0)) - r\InDegree \cdot \sum_{t = 1}^{s - 1} h(e_s \mid \sigma_{s - 1})\\
	={}&
	f(\sigma^*) - r\InDegree \cdot \sum_{t = 1}^{s - 1} h(e_s \mid \sigma_{s - 1})
	\geq
	f(\sigma^*) - r\InDegree \cdot \sum_{t = 1}^{s - 1} [f(\sigma_t) - f(\sigma_{t - 1})]\\
	={} &
	f(\sigma^*) - r\InDegree \cdot [f(\sigma_{s - 1}) - f(\sigma_0)]
	\geq
	f(\sigma^*) - r\InDegree \cdot f(\sigma_{s - 1})
	\enspace.
\end{align*}
The first equality follows since the fact that $\sigma_0$ is empty implies $E(\sigma_0) = \varnothing$ and $E(\sigma^*) \cap \cE_1 = E(\sigma^*)$. 
Additionally, the second inequality follows from Observation~\ref{obs:gain}, and the last inequality follows from the non-negativity of $f$. This proves that the first inequality of the lemma holds for every $1 \leq s \leq \ell'$. In the rest of the proof we aim to prove the second inequality, and thus, assume $1 \leq s \leq \ell$.

Recall now that $e_s$ is the edge of $\cE_s$ maximizing $h(e_s \mid E(\sigma_{s - 1}))$ and that the size of $E(\sigma^*) \cap \cE_s \subseteq E(\sigma^*)$ is at most $\InDegree k$ by Observation~\ref{obs:size}. Thus, by the submodularity of $h$,
\begin{align*}
	h(e_s \mid E(\sigma_{s - 1}))
	\geq{} &
	\frac{\sum_{e \in E(\sigma^*) \cap \cE_i} h(e \mid E(\sigma_{s - 1}))}{|E(\sigma^*) \cap \cE_s|}
	\geq
	\frac{h(E(\sigma^*) \cap \cE_s \mid E(\sigma_{s - 1}))}{\InDegree k}\\
	\geq{} &
	\frac{h((E(\sigma^*) \cap \cE_s) \cup E(\sigma_{s - 1})) - h(E(\sigma_{s - 1}))}{\InDegree k}\\
	\geq{} &
	\frac{\{f(\sigma^*) - r\InDegree \cdot [f(\sigma_{s - 1}) - f(\sigma_0)]\} - f(\sigma_{s - 1})}{\InDegree k}
	\geq
	\frac{f(\sigma^*) - (r\InDegree + 1) \cdot f(\sigma_{s - 1})}{\InDegree k}
	\enspace.
\end{align*}
The second inequality of the lemma now follows by combining the last inequality with Observation~\ref{obs:gain}.
\end{proof}

\begin{corollary} \label{cor:final_raw_hyper}
If $r\InDegree + 1 < \InDegree k$, then $f(\sigma_{\ell}) \geq \frac{f(\sigma^*)}{r\InDegree + 1} + \frac{[1 - (r\InDegree + 1) / (\InDegree k)]^\ell}{r\InDegree + 1} \cdot [(r\InDegree + 1)f(\sigma_0) - f(\sigma^*)]$.
\end{corollary}
\begin{proof}
To prove the corollary, we prove by induction the stronger claim that, for every $0 \leq s \leq \ell$,
\[
	f(\sigma_s)
	\geq
	\frac{f(\sigma^*)}{r\InDegree + 1} + \frac{[1 - (r\InDegree + 1) / (\InDegree k)]^s}{2r\InDegree + 1} \cdot [(r\InDegree + 1)f(\sigma_0) - f(\sigma^*)]
	\enspace.
\]
For $s = 0$ this inequality is true since
\begin{align*}
	f(\sigma_0)
	={} &
	\frac{f(\sigma^*)}{r\InDegree + 1} + \frac{1}{r\InDegree + 1} \cdot [(r\InDegree + 1)f(\sigma_0) - f(\sigma^*)]\\
	={} &
	\frac{f(\sigma^*)}{r\InDegree + 1} + \frac{[1 - (r\InDegree + 1) / (\InDegree k)]^0}{r\InDegree + 1} \cdot [(r\InDegree + 1)f(\sigma_0) - f(\sigma^*)]
	\enspace.
\end{align*}
Assume now that the claim holds for $s - 1 \geq 0$, and let us prove it for $s$. By Lemma~\ref{lem:iteration_guarantee_hyper},
\[
	f(\sigma_s)
	\geq
	f(\sigma_{s - 1}) + \frac{f(\sigma^*) - (r\InDegree + 1) \cdot f(\sigma_{s - 1})}{\InDegree k}
	=
	\left(1 - \frac{r\InDegree + 1}{\InDegree k}\right) \cdot f(\sigma_{s - 1}) + \frac{f(\sigma^*)}{\InDegree k}
	\enspace.
\]
Plugging in the induction hypothesis, we get
\begin{align*}
	f(\sigma_i)
	\geq{} &
	\left(1 - \frac{r\InDegree + 1}{\InDegree k}\right) \cdot \left\{\frac{f(\sigma^*)}{r\InDegree + 1} + \frac{[1 - (r\InDegree + 1) / (\InDegree k)]^{s - 1}}{r\InDegree + 1} \cdot [(r\InDegree + 1)f(\sigma_0) - f(\sigma^*)]\right\} \\&+ \frac{f(\sigma^*)}{\InDegree k}
	=
	\frac{f(\sigma^*)}{r\InDegree + 1} + \frac{[1 - (r\InDegree + 1) / (\InDegree k)]^s}{r\InDegree + 1} \cdot [(r\InDegree + 1)f(\sigma_0) - f(\sigma^*)]
	\enspace.
	\qedhere
\end{align*}
\end{proof}

We are now ready to prove Theorem~\ref{thm:hyper}.
\begin{proof}[Proof of Theorem~\ref{thm:hyper}]
First, we need to consider the case that Algorithm~\ref{alg:greedy_path_hyper} terminates because the set $\cE$ becomes empty. In this case $\cE_{\ell'} = \varnothing$, which implies
\[
	h((E(\sigma^*) \cap \cE_{\ell'}) \cup E(\sigma_{\ell' - 1}))
	=
	h(E(\sigma_\ell))
	=
	f(\sigma_\ell).
\]
Using Lemma~\ref{lem:iteration_guarantee_hyper} for $s = \ell' = \ell + 1$, this observation implies
\[
	f(\sigma_\ell)
	\geq
	f(\sigma^*) - r\InDegree \cdot f(\sigma_{\ell})
	\Rightarrow
	f(\sigma_{\ell})
	\geq
	\frac{f(\sigma^*)}{r\InDegree + 1}
	\enspace,
\]
which proves the theorem. Thus, in the rest of the proof we may assume that Algorithm~\ref{alg:greedy_path_hyper} terminated because $\sigma$ reached a size larger than $k - r$.

Consider now the case that $r\InDegree + 1 \geq \InDegree k$. In this case
\begin{align*}
	f(\sigma_{\ell})
	\geq{} &
	f(\sigma_1)
	=
	f(\sigma_0) + [f(\sigma_1) - f(\sigma_0)]
	\geq
	f(\sigma_0) + \frac{f(\sigma^*) - f(\sigma_0)}{\InDegree k}\\
	\geq{} &
	\frac{f(\sigma^*)}{\InDegree k}
	\geq
	\frac{f(\sigma^*)}{r\InDegree + 1}
	\geq
	\frac{1 - e^{-(1 - \frac{r}{k})}}{r\InDegree + 1} \cdot f(\sigma^*)
	\enspace,
\end{align*}
where the first inequality holds since $\sigma_1$ is a prefix of $\sigma_\ell$ and the second inequality follows from Lemma~\ref{lem:iteration_guarantee_hyper}. Thus, it remains to prove the theorem in the more interesting case of $r\InDegree + 1 < \InDegree k$.

%
Observe that
\[
	[1 - (r\InDegree + 1) / (\InDegree k)]^\ell
	\leq
	e^{-(r + 1 / \InDegree) \cdot (\ell / k)}
	\enspace.
\]
Plugging this inequality into Corollary~\ref{cor:final_raw_hyper} gives
\begin{align*}
	f(\sigma_{\ell})
	\geq{} &
	\frac{f(\sigma^*)}{r\InDegree + 1} + \frac{[1 - (r\InDegree + 1) / (\InDegree k)]^\ell}{r\InDegree + 1} \cdot [(r\InDegree + 1)f(\sigma_0) - f(\sigma^*)]\\
	\geq{} &
	\frac{1 - [1 - (r\InDegree + 1) / (\InDegree k)]^\ell}{r\InDegree + 1} \cdot f(\sigma^*)
	\geq
	\frac{1 - e^{-(r + 1 / \InDegree) \cdot (\ell / 	k)}}{r\InDegree + 1} \cdot f(\sigma^*)
	\enspace.
\end{align*}
At this point we need a lower bound on $\ell$. One can note that $|\sigma|$ starts as $0$, increases by at most $r$ in each iteration of the loop of Algorithm~\ref{alg:greedy_path_hyper} and ends up with a value of at least $k - r + 1$ by our assumption. Thus, the number $\ell$ of iterations must be at least $(k - r + 1)/r$. Plugging this observation into the previous inequality gives
\[
	f(\sigma_{\ell})
	\geq
	\frac{1 - e^{-(r + 1 / \InDegree) \cdot (1 - (r - 1) / k)/r}}{r\InDegree + 1} \cdot f(\sigma^*)
	\geq
	\frac{1 - e^{-(1 - r / k)}}{r\InDegree + 1} \cdot f(\sigma^*)
	\enspace.
	\qedhere
\]
\end{proof}

\section{Applications} \label{apps}
\subsection{Movie Recommendation}\label{moviesec}

In this application, we use the \emph{Movielens 1M} dataset~\citep{movielens1m} to recommend movies to users based on the films they have reviewed in the past. This dataset contains 1,000,209 anonymous, time-stamped ratings made by 6,040 users for 3,706 different movies. As in \cite{seq17}, we do not want to predict a user's rating for a given movie, instead we want to predict which movies the user will review next. 

One issue with this dataset is that the distribution of the number of ratings per user (shown in Figure~\ref{userDist1}) has a very long tail, with the most prolific reviewer having reviewed 2,314 movies. In order for our data to be representative of the general population, we remove all users who have rated fewer than 20 movies or more than 50 movies. We also remove all movies with fewer than 1,000 reviews. This leaves us with 67,757 ratings made by 2,047 users for 207 different movies. 

We first group and sort all the reviews by user and time-stamp, so that each user $i$ has an associated sequence $\sigma^i$ of movies they have rated, where $\sigma^i_j$ refers to the $j^{th}$ movie that user $i$ has reviewed. We use a 90/10 training/testing split of the data and 10-fold cross validation.

For each user $i$ in the test set ($D_{test}$), we use their first 8 movies as a given starting sequence $S_i = \{ \sigma_1^i \hdots \sigma_8^i \}$. We want to use $S_i$ to select $k$ movies that we think user $i$ will review in the future. Therefore, for each user $i$, we build a hypergraph $H_i = (V,E_i)$, where $V = \{ v_1, \hdots, v_n \}$ is the set of all movies, and $E_i$ is a set of hyperedges. Each hyperedge $e_s$ has value $p_s$, where $s$ is a movie sequence of length at most 3. Intuitively, $p_s$ is the conditional probability of reviewing the last movie in $s$ given that the rest of the movies in $s$ have already been reviewed in the proper order.

Since we use empirical frequencies in the training data to calculate these conditional probabilities, we may run into the issue of overfitting to rare sequences. To avoid this, we add a parameter $d$ to the denominator of our calculation of each edge value. This will increase the relative value for sequences that appear more often. In this experiment, we use $d=20$.

More formally, define $N_s$ to be the number of users in the training set ($D_{train}$) that have reviewed all the movies in the sequence $s$ in the proper order. Also define $s_l$ to be last element in $s$, and $s'$ to be $s$ with $s_l$ removed. Now we can define the value of each edge $e_s$ as follows:

\begin{equation} 
\label{cp}
p_s = \begin{cases} 
      \displaystyle \frac{N_s}{N_{s'} + d} & s' \subseteq S_i \enspace, \vspace{0.1in}\\
      \displaystyle p_{s'} \frac{N_s}{N_{s'} + d} & \text{otherwise} \enspace.
   \end{cases}
\end{equation}

As mentioned above, the idea is that $p_s$ represents the conditional probability of reviewing $s_l$ given that all the movies in $s'$ have already been reviewed in the proper order. If user $i$ has not reviewed all the movies in $s'$, then we scale down the value of that edge by $p_{s'}$ (i.e., the conditional probability of reviewing all the movies in $s'$).

Note that if $s' = \varnothing$, then we define $N_{s'} = |D_{train}|$, thus ensuring that this definition also applies for self-loops. A small subgraph of a fully trained hypergraph is shown in Figure \ref{subgraph}.

We use a probabilistic coverage utility function as our non-negative monotone submodular function $h$. Mathematically, 

\[
h(E) = \sum_{v \in nodes(E)} \Big[ 1 - \prod_{s \in E \mid s_l = v} (1 - p_{s}) \Big]
\]

We compare the performance of our algorithms, Sequence-Greedy and Hyper Sequence-Greedy, to the existing submodular sequence baseline (OMegA), as well as a naive baseline (Frequency), which just outputs the most popular movies that the user has not yet reviewed. 

\begin{figure*}[t]
\centering
\subfloat[]{\includegraphics[height=1.9in]{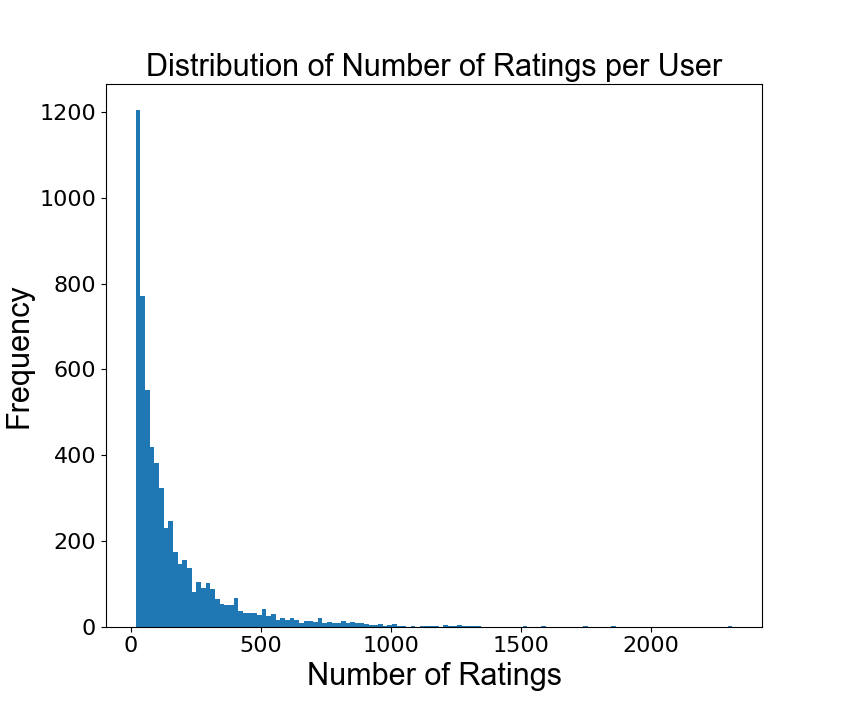}\label{userDist1}}
\hspace{-0.11in}
\subfloat[]{\includegraphics[height=1.8in]{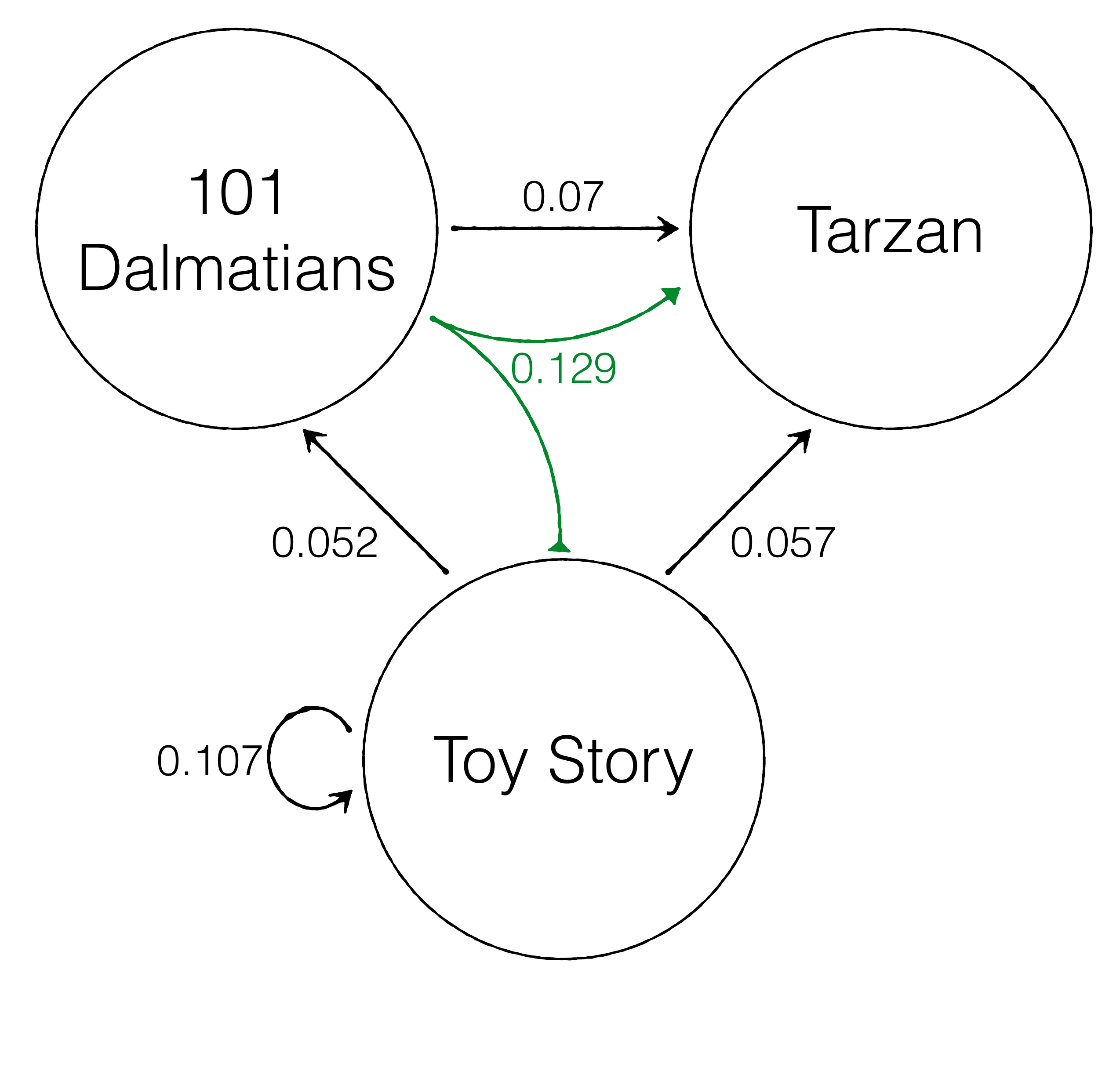}\label{subgraph}}
\subfloat[]{\includegraphics[height=1.9in]{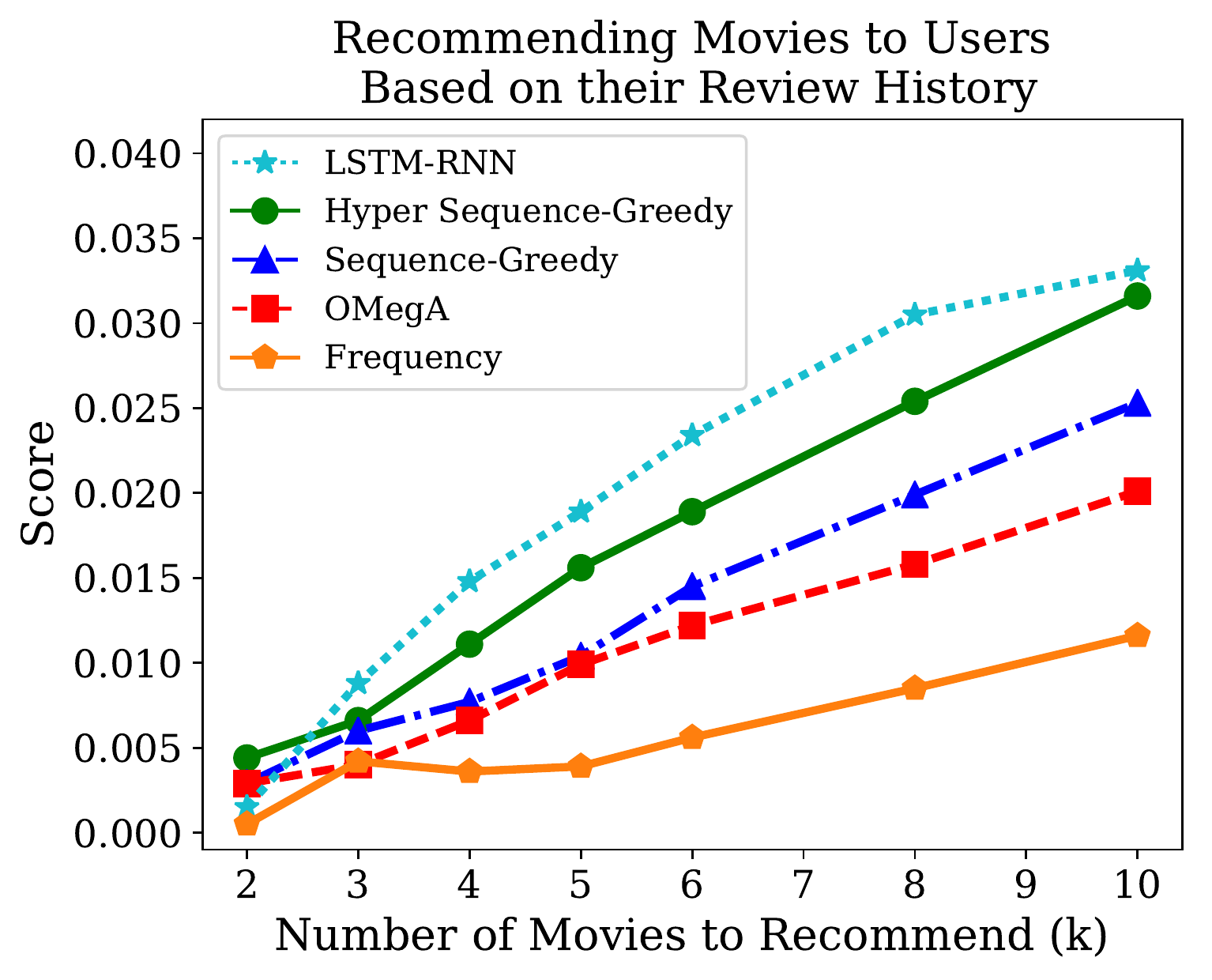}\label{movieGraph}}
\caption{(a)~Shows the long-tailed distribution of the number of ratings per user in the \emph{Movielens 1M} dataset. (b)~Shows a small subgraph of the overall hypergraph $H$ that we train. For clarity, we only show edges with value $p_s > 0.05$, as defined in equation~$(\ref{cp})$. We also highlight the size 3 hyperedge in green. (c)~Shows the performance of our algorithms against existing baselines for various cardinalities $k$.}
\label{fig:movies}
\end{figure*}

We also compare to a simple long short-term memory (LSTM) recurrent neural network (RNN). In addition to tuning parameters, we experimented with various frameworks such as training on uniform vs. variable-sized sequences. In the end, we obtained the best results when we trained the neural network on the first $k$ movies of each $\sigma^{i}$, where the target is to predict the next $k$ movies that the user $i$ will review. In terms of the architecture, we use one layer of 512 LSTM nodes (with a dropout of 0.5) followed by a dense layer with a softmax activation that returns a  $207 \times 1$ vector $P$, where entry $P_i$ is the probability that movie $i$ will be reviewed. For each $k$, we simply return the $k$ highest values in $P$. As before, we used a 90/10 training/testing split with 10-fold cross validation.

We hypothesize that, with enough data, neural networks will outperform our algorithms. However, with this comparison, we would like to show that in situations where data is relatively scarce, our algorithms are competitive with existing neural network frameworks, while also providing theoretical guarantees.

To measure the accuracy of a prediction, we use a modified version of the Kendall tau distance \citep{kendallTau}. First, for any sequence $\sigma$, we define $T(\sigma)$ to be the set of all ordered pairs in $\sigma$. For example, if $\sigma = \{1,3,2\}$, then $T(\sigma) = \big[ (1,3),(1,2),(3,2) \big]$. 


Let $P_i$ be our predicted sequence for the next $k$ movies that user $i$ will review, and let $Q_i$ be the next $k$ movies that user $i$ actually reviewed. Then, we define the accuracy of the prediction $P_i$ as follows.

\[
\tau(P_i,Q_i) = \frac{|T(P_i) \cap T(Q_i)|}{|T(Q_i)|}
\]

In other words, $\tau(P_i,Q_i)$ is the fraction of ordered pairs of the true answer $Q_i$ that appear in our prediction $P_i$. Our experimental results in terms of this accuracy measure are summarized in Figure \ref{movieGraph}.

These results showcase the power of using hypergraphs, as Hyper Sequence-Greedy consistently outperforms Sequence-Greedy. We also notice that Hyper Sequence-Greedy outperforms the score of the existing baseline OMegA by roughly 50\%.

\vspace{-0.1in}
\subsection{Online Link Prediction}\label{wikisec}

In this application, we consider users who are searching through Wikipedia for some target article. Given a sequence of articles they have previously visited, we want to predict which link they will follow next. We use the Wikispeedia dataset~\citep{wikispeedia}, which consists of 51,138 completed search paths on a condensed version of Wikipedia that contains 4,604 articles and 119,882 links between them. 

The setup for this problem is similar to that of section \ref{moviesec}, so we will only go over the main differences. Again we will use a 90/10 training/testing split of the data with 10-fold cross validation.

For each training set $D_{train}$, we build the underlying hypergraph $H = (V,E)$. This time, $V$ is the set of all articles, and $E$ is a set of hyperedges $e_s$ where $p_s$ is the conditional probability of moving to article $s_l$ given that the user had just visited $s'$ in succession.

\begin{figure}[h]
\vspace{0.15in}
\begin{center}
\includegraphics[scale = 0.24]{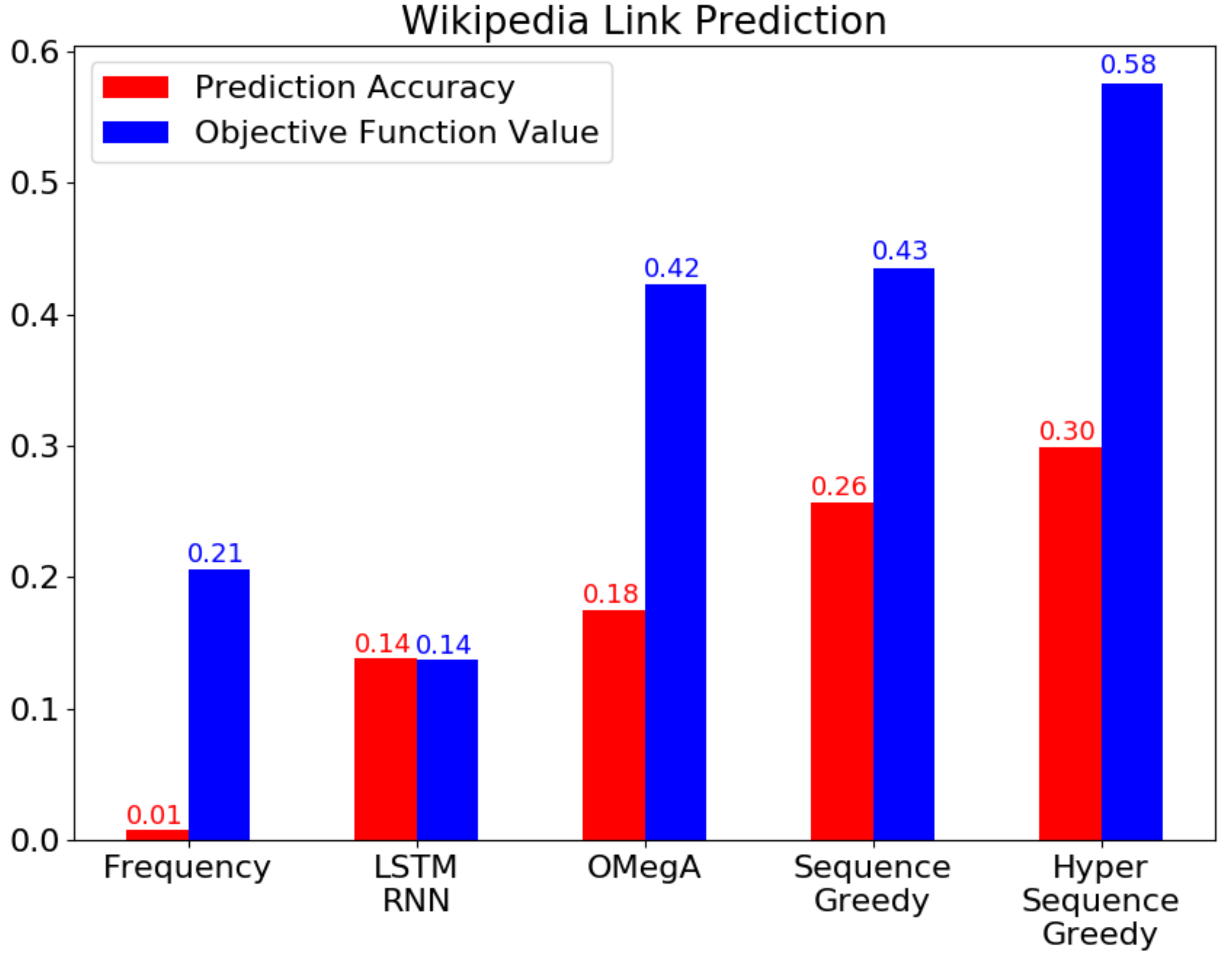}
\end{center}
\vspace{-.05in}
\caption{Given a sequence of articles a Wikipedia user has visited, we want to predict the next link they will click. This bar chart shows the prediction accuracy, as well as the objective function value, of various algorithms.}
\label{wikiGraph}
\end{figure}

\newpage
For each testing set $D_{test}$, we will use the last article in each completed path as the target, and the previous 3 articles as the given sequence. This means we will be able to use hyperedges of up to size 4. We employ the same probabilistic coverage function $h$ and the same baseline comparisons as in Section \ref{moviesec}. For this application, our neural network was most effective when we used a single layer of 32 LSTM nodes (with a dropout of 0.2). Our results are shown in Figure~\ref{wikiGraph}.

In this case, Hyper Sequence-Greedy exhibits the best performance. 
We see that the simple neural network implementation is outperformed by Hyper Sequence-Greedy as well as by some of the baselines. This is likely a result of the data in this experiment being more sparse. Although in this application we technically have more data than in the previous one, here we attempt to choose between 4,604 articles, rather than just 207 movies.

We also show the results that the various algorithms achieve when evaluated on our objective function $f(\sigma) = h\big( E(\sigma) \big)$. 
Asides from the LSTM-RNN, which doesn't consider the objective function at all, we see that the objective function values are relatively in line with the prediction accuracy. This demonstrates that the probabilistic coverage function was a good choice for the objective function.

\subsection{Course Sequence Design}

In this final application we want to use historical enrollment data in Massive Open Online Courses (MOOCs) to generate a sequence of courses that we think would be of interest to users. We use a publicly available dataset \citep{edx14} that covers the first year of open online courses offered by edX. The dataset consists of 641,139 registrations from 476,532 unique users across 13 different online courses offered by Harvard and MIT. Amongst a plethora of other statistics, the data contains information on when each user first and last accessed each course, how many course chapters they accessed, and the grade they achieved if they were ultimately certified (i.e., fully completed) in the course. 


One natural way to think about the value of a sequence of courses is in terms of prerequisites. That is, in what order should we offer courses to students in order to help them learn as much as possible. This model comes with a natural measure of success as well, which is the grade each student gets in each course. Unfortunately, out of the 476,532 unique users in this dataset only 180 were certified (and thus, received grades) in 3 or more courses. Furthermore, this dataset only contains 13 different courses (shown in Figure \ref{teach1}), none of which are logical prerequisites for each other.

Instead, we can think about a sequence of courses being valuable if they will all be interesting to a user who registers for them. Similarly to the prerequisites model where the order of courses affects the user's grade, the order in which a user registers for courses should also affect their interest. In this dataset, we can measure interest by the percentage of the course that the user accessed. In particular, we say that if a user was interested in a course $i$ if she accessed at least one-third of all the chapters for course $i$. 

As always, we need to build the underlying hypergraph $H = (V,E)$ for each training set. In this case, $V$ is the set of all courses and $E$ is a set of hyperedges of form $e_s$, where $s$ is a sequence of at most 3 courses and $p_s$ is the probability that a user will be interested in $s_l$ given that she previously showed interest in $s'$ in the proper order. Recall that $s_l$ is the last course in $s$, and $s'$ is the sequence obtained from $s$ after deleting $s_l$. As in section \ref{moviesec}, we also use a parameter $d$ to avoid overfitting to rare sequences. In this case we use $d = 100$.
However, unlike Section \ref{moviesec}, we are not making recommendations based on a user's history. Instead each algorithm will use the underlying hypergraph to build a single sequence $\sigma$. Since we are not starting with any given sequence, we can finally run Sequence-Greedy and Hyper Sequence-Greedy both forwards and backwards, and take the maximum of the two results.

Different users will naturally have different interests, so it is unreasonable to expect that any single sequence $\sigma$ will work for all users. However, if $\sigma$ is a ``good" sequence, we could expect that users who start all the courses in $\sigma$ in the correct order ultimately end up showing interest in those courses. Intuitively, the idea is that $\sigma$ should capture a sequence of courses with some common theme and present them in the best possible order. Therefore, if a user begins all the courses in $\sigma$ they likely have some interest in this common theme. Hence, if $\sigma$ is a good sequence, it will present these courses in a good order and properly pique the interest of these users. 

Mathematically, we define $S_{\sigma}$ to be the set of users who started all the courses in $\sigma$ in the proper order, and $c_{ij}$ to be the percentage of course $j$ that user $i$ completed. Therefore, the value of $\sigma$ for a given test set $D_{test}$ is defined as: 
\vspace{-0.05in}
\[
D_{test}(\sigma) = \displaystyle \frac{ \displaystyle \sum_{i \in S_{\sigma}} \sum_{j \in \sigma} c_{ij}}{\big( |S_{\sigma}| +d \big) |\sigma| }
\]

Using a 75/25 training/testing split of the data and 4-fold cross validation, we compare the effectiveness of Hyper Sequence-Greedy, Sequence Greedy, OMegA, and Frequency for the task of selecting a sequence of 4 courses. 
Note that due to the inherent randomness in the training/testing split, there is some variance in the results. To be conservative, the results shown in Figure \ref{teach2} are actually on the lower end of the performance we see from our algorithms. Figure \ref{teach3} shows some representative sequences.

\begin{figure*}[tbp]
\centering
\subfloat[]{\includegraphics[height=1.72in]{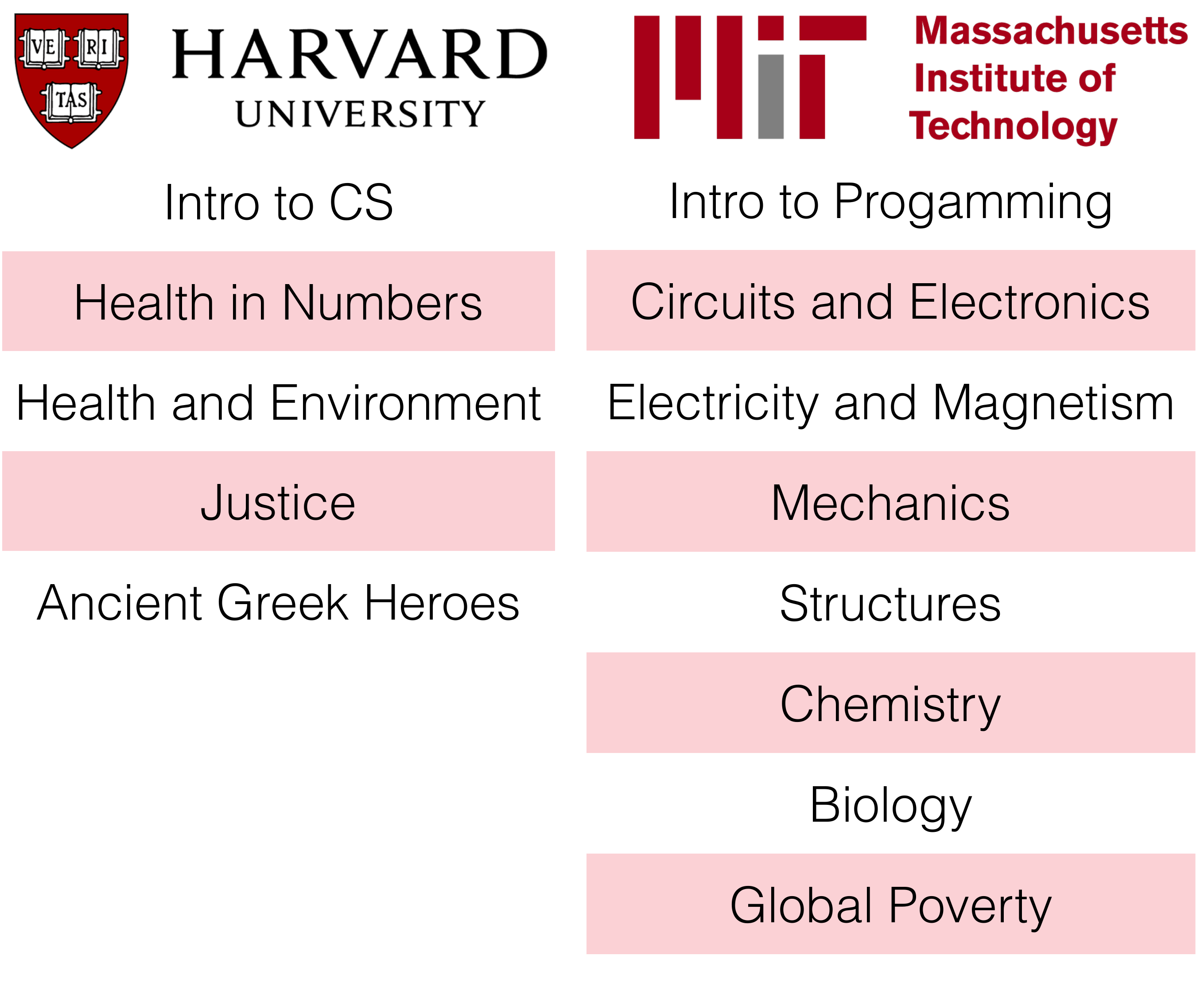}\label{teach1}}
\hspace{0.05in}
\subfloat[]{\includegraphics[height=1.72in]{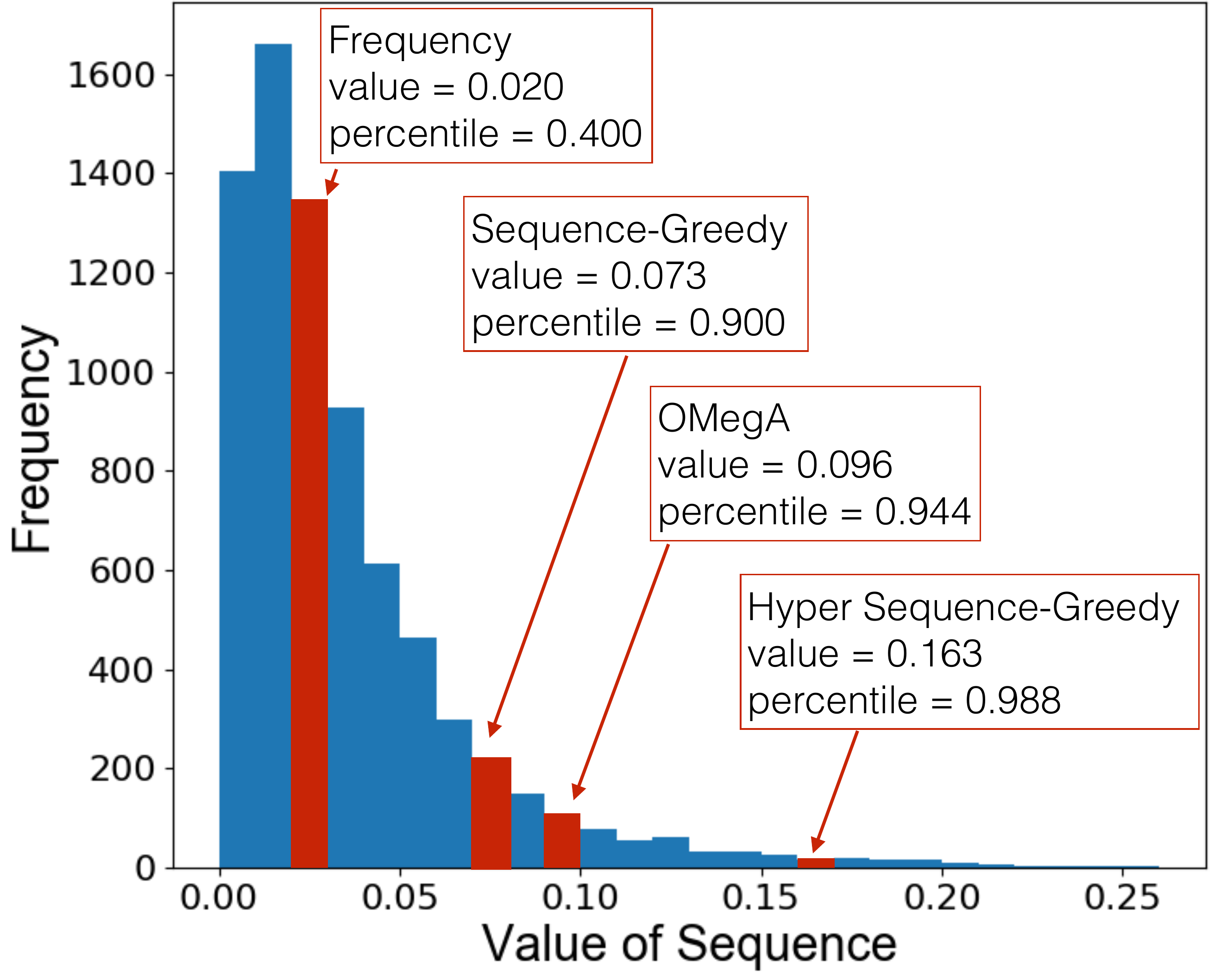}\label{teach2}}
\hspace{0.05in}
\subfloat[]{\includegraphics[height=1.72in]{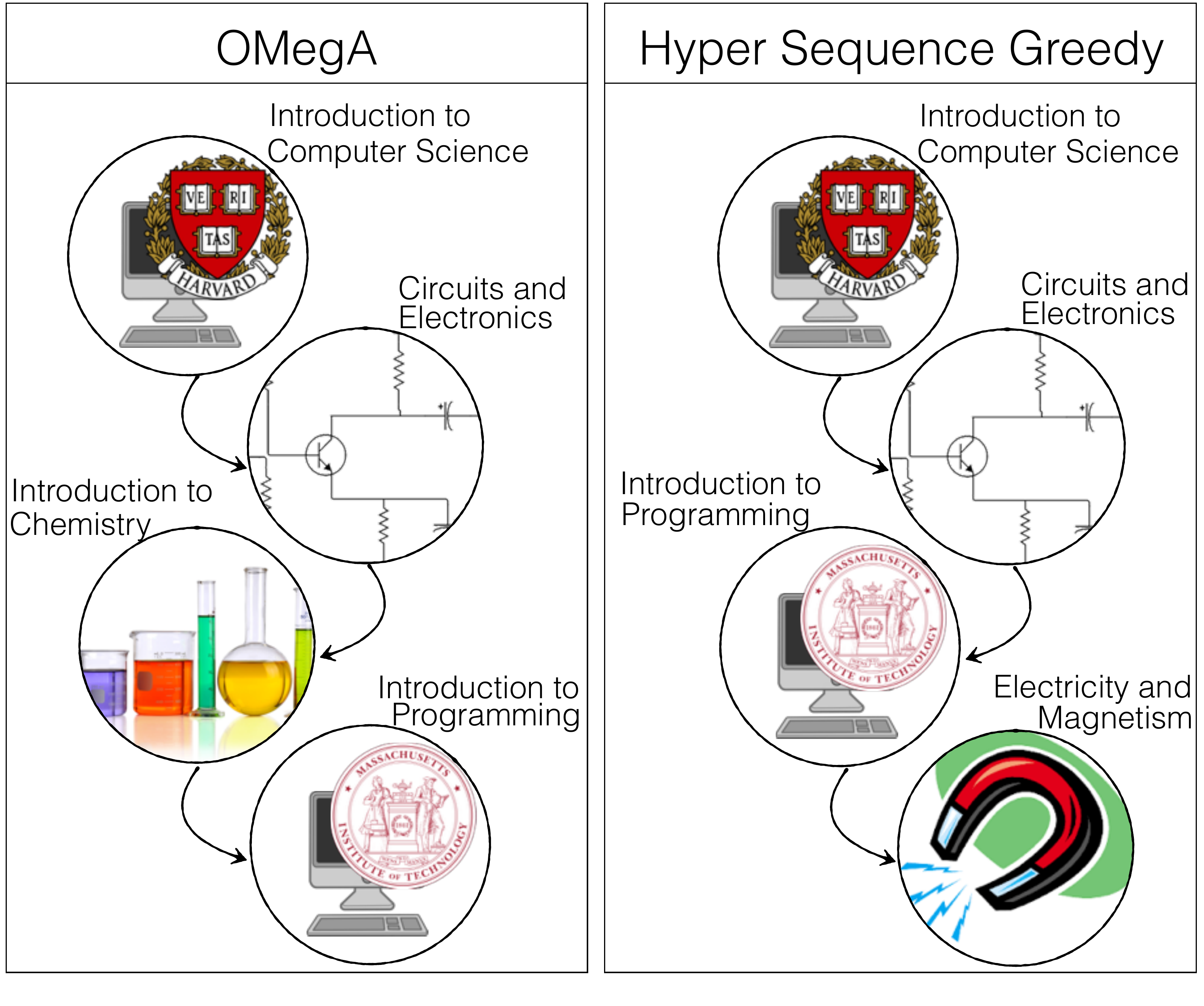}\label{teach3}}
\caption{(a) shows the 13 different courses that were available to students in this dataset. (b) is a histogram of the value of every 4 course sequence that appears in the dataset. The values of the courses selected by the various algorithms are overlayed on top of the corresponding bar in the histogram. (c) shows representative course sequences selected by OMegA and Hyper Sequence-Greedy.}
\label{fig:teach}
\end{figure*}

We see that Hyper Sequence-Greedy outperforms the other algorithms, as expected. From the histogram, we also see that Hyper Sequence-Greedy tends to select one of the best possible sequences, with Sequence-Greedy and OMegA both performing in the $90^{th}$ percentile.  Somewhat surprisingly, OMegA (which has to use a random topological order in the absence of a directed acyclic graph) outperforms Sequence-Greedy. However, this may be explained by the fact that $k=4$ is relatively small. Unfortunately, only 1,153 users even started more than 4 courses, meaning that we cannot effectively test sequences of larger length with this dataset. 

\section{Conclusion}

This paper extended results on submodular sequences from directed acyclic graphs to general graphs and hypergraphs. Our theoretical results showed that both our algorithms, Sequence-Greedy and Hyper Sequence-Greedy, approach a constant factor approximation to the optimal solution (for constant $\Delta$).  Furthermore, we demonstrated the utility of our algorithms, in particular the power of using hyperedges, on real world applications in movie recommendation, online link prediction, and the design of course sequences for MOOCs.\\

\textbf{Acknowledgements} We acknowledge support from DARPA YFA (D16AP00046), AFOSR YIP (FA9550-18-1-0160), ISF grant 1357/16, and ERC StG SCADAPT.

%
%


\newpage
\bibliographystyle{abbrvnat}
\bibliography{references-sub}

%

\end{document}